\documentclass[preprint,12pt,authoryear]{elsarticle}

\usepackage{amssymb}
\usepackage{amsthm}
\usepackage{tikz-cd} 
\usepackage{comment}
\usepackage{amsmath,amssymb}
\usepackage{yfonts}
\usepackage{tikz}
\usepackage{tikz-cd} 

\newtheorem{theorem}{Theorem}

\newcommand{\matter}{\mathcal{M}}
\newcommand{\spa}{\mathcal{S}}
\newcommand{\conf}{\mathcal{C}}
\newcommand{\deform}{\mathcal{D}}

\journal{Journal of Geometry and Physics}

\begin{document}
\title{The Principle Bundle Structure\\ of \\ Continuum Mechanics}

\author{Stefano Stramigioli, \\ University of Twente, The Netherlands}

\begin{abstract}
In this paper it is shown that the structure of the configuration space of any continua is what is called in differential geometry a {\it principle bundle} \cite{Frankel2011ThePhysics}. A principal bundle is a structure in which all points of the manifold (each configuration in this case) can be naturally projected to a manifold called the {\it base manifold}, which in our case represents pure deformations. All configurations projecting to the same point on the base manifold (same deformation) are called fibers. Each of these fibers is then isomorphic to the Lie group $\mathfrak{se(3)}$ representing pure rigid body motions.  Furthermore, it is possible to define what is called a connection and  this allows to split any continua motion in a rigid body sub-motion and a deformable one in a completely coordinate free way. As a consequence of that it is then possible to properly define a pure deformation space on which an elastic energy can be defined. This will be shown using screw theory \cite{Ball:1900}, which is vastly used in the analysis of rigid body mechanisms but is not normally used to analyse continua. Beside the just mentioned result, screw theory will also be used to relate concepts like helicity and enstrophy to screw theory concepts.
\end{abstract}

\maketitle

\section{Introduction}
The study of continuous matter is of fundamental importance in all branches of engineering mechanics \cite{book:abrahammarsden}, spanning from solid deformations \cite{MarsdenHughes1983} to fluids \cite{Chorin1993AMechanics}.
A lot of interest is present also about proper coordinate free analysis descriptions of continua, like for example the excellent recent work of \cite{Objrates2021}.

On the other hand in rigid body mechanics, the {\it theory of screws}, \cite{Ball:1900} is vastly used to describe the motion of mechanisms of interconnected rigid bodies. A nice example of its use can for example be found in the analysis of dynamic balance by \cite {deJong2019ABalance}. The theory of screws is based on two theorems: Mozzi's theorem, for which an hystorical perspective can be found in 
\cite{Ceccarelli2000}\footnote{In the literature of screw theory Chasles theory (1830) is often cited but Mozzi's presented the same result in 1763.} and Poinsot's theorem. Mozzi's theorem states that any instantaneous rigid body motion can be interpreted as the action of a rotation around an axis in space superimposed to a possible motion along the same axis, proportional to a scalar called the {\it pitch} which indicates the length traveled along the line for a full rotation along the line, like the pitch of a ``workshop'' screw. Poinsot's theorem is a dual theorem which basically says that any systems of forces applied to a rigid body, would have the same effect of a single force applied along a line in space and superimposed to a possible torque around the same line, also relating the two by a scalar pitch. These two theorems where then used in \cite{Ball:1900} to create a theory based on the geometry of lines, or better screws, for the analysis of rigid body kinematics. An extensive nice treatment of the mathematics and its relation to Lie groups and Clifford's algebras can be found in  \cite{book:selig}, and specifically to Lie algebras in \cite{JournalBall2002}. 
In this work it will be shown that it is possible to associate to each velocity and vorticity in a certain point of a continua, an infinitesimal (associated to the infinitesimal volume at the point) screw describing them at the same time. Such an infinitesimal screw (relating to the infinitesimal volume element with velocity $v$ and vorticity $dv$), as in the rigid body case, will be characterised by an axis and a pitch, but in this case the motion of the point will be considered acted upon by the infinitesimal screw. It will be shown that the concepts of helicity and enstrophy, well known in fluid dynamics, will get a very geometric interpretation in the context of infinitesimal screws.

Considering that each of these screws are equivalent to elements of a Lie algebra \cite{JournalBall2002}, they can be summed or integrated. Such an operation will reveal the superimposed motion that each point will generate: if all motions would follow the same rigid body motion, the resultant would be such a motion, but in case of non rigid motions, the resultant will show the ``main rigid body motion'' of the continuum. This construction will be presented precisely and it will be shown that, thanks to this, it will be possible to uniquely decompose, in a completely coordinate free way, the rigid component of a continuum motion and the resting pure deformation. This will be done by the introduction of what is called a differential geometric connection  \cite{Frankel2011ThePhysics} in the principal bundle of motions. 

During the paper, to stress the coordinate invariance of the approach, different terms will be used as synonymous like geometric or intrinsic.

The paper will start in Sec.\ref{sec:setting} with the general framework formulation of the analysis of a continua. In Sec.\ref{sec:princbundle} it will be shown that the configuration space of a continuum has indeed a principle bundle structure whose fibers are isomorphic to $SE(3)$. In Sec. \ref{sec:screws} the geometry and basics of screws theory will be reviewed to then in Sec.\ref{sec:infscrews} use this insight to define the connection on the principle bundle introduced in Sec.\ref{sec:princbundle}. Finally, in Sec.\ref{sec:implications} the results will be framed in the context of a geometric description of potential energy of the deformation, only using an energy formulation to then draw some concluding remarks in Sec.\ref{sec:conclusions}.

\paragraph*{Notation} 

A compact, orientable, $n$-dimensional Riemannian manifold $M$ with (possibly empty) boundary $\partial M$ models the spatial container of a non relativistic continuum mechanical system. It possesses a metric field $g$ and induced Levi-Civita connection $\nabla$ . The space of vector fields on $M$ 
is defined as
the space of sections of the tangent bundle $TM$, that will be denoted by $\Gamma({TM})$. The space of differential $p$-forms is denoted by $\Omega^p(M)$ and we also refer to $0$-forms as functions, $1$-forms as co-vector fields and $n$-forms as top-forms. For any
$v\in \Gamma(TM)$, we use the standard definitions for the interior product by $\iota_v:\Omega^p(M) \to \Omega^{p-1}(M)$ and the Lie derivative operator $\mathcal{L}_v$ acting on tensor fields of any valence. 
The Hodge star operator
$\star:\Omega^p(M) \to \Omega^{n-p}(M)$, the volume form $\mu=\star 1$, as well as the "musical" operators $\flat:\Gamma(TM) \to \Omega^{1}(M)$ and $\#:\Omega^{1}(M) \to \Gamma(TM)$, which respectively transform vector fields to 1-forms and vice versa, are all uniquely induced by the Riemannian metric in the standard way. 
 To complement a purely coordinate-free notation, we will represent tensorial quantities of interest in local coordinates, denoted by $x$, when considered insightful, using Einstein summation convention.

\section{The Continuum Mechanics Setting \label{sec:setting}}
The goal of this paper is to show the intrinsic structures which appear in the description of the kinematics of a continua. For this reason, there will be no use of coordinates. The setting which will be used, using the approach of \cite{Noll1978} is the one of a continua moving in space and the various entities of the problem will be introduced next.
\subsection{Matter and Space}
We consider the motion of {\it matter} in {\it space} where matter is mathematically modeled as a manifold $\matter$ of dimension $n$ to which some properties may be associated via sections of a proper bundle. For example, in the case of a purely mechanical description, we can consider $\mu_m \in \Omega^n(\matter)$ as the mass top form describing the mass distribution, where we have indicated with $\Omega^n(\matter)$ the sections  of the $n$ alternating bundle to which  $n$-forms on $\matter$ belong. Any other property could be defined in the same way. What is important to realise is that such properties are all those properties which are conserved for physical reasons and will be advected by the motion because strictly connected to the matter and matter only, independently of any motion it takes in the space. This could be for example also electrical charge. On the other hand the scalar density $\rho$ is a quantity which is not conserved/advected because its definition does not depend only on the matter, but on the relation between the mass topform $\mu_m$ and a representation of the volume form of the space.

We consider {\it space} as an $n$ dimensional  Riemannian manifold $(\spa, g)$ where with $g \in \text{sym}_+(T_2^0 \spa)$ we indicate the positive definite symmetric metric field on $\spa$ with induced  volume form $\mu \in \Omega^n(\spa)$. 

\subsection{The configuration space}
We can now put {\it matter} into  {\it space} with an embedding of the form:
\[
e:\matter \rightarrow \spa
\]
\noindent which will describe in a certain moment where each of the points of matter will be situated in space. Each of such objects will define therefore a configuration of the body in space and therefore the set  $\conf$ of all proper embedding of $\matter$ into $\spa$ will describe the configuration space. It is possible to see that, inheriting the differential structures of $\matter$ and $\spa$, $\conf$ is an infinite dimensional differentiable manifold. From now on we will consider furthermore the elements of $\conf$ to be bijections to simplify the exposition of what follows. This is always possible if we generalise the concept of matter to also vacuum matter representing those parts of space where no mass is present but a more precise description can be used \cite{Noll1978,Objrates2021}.

\subsection{Motions}
We can now define a motion as a smooth curve in $\conf$ of the type
\[
c: I \rightarrow \conf \; : t \mapsto c_t 
\]
\noindent which associates to a real time value $t$ in the interval $I$ the configuration $c(t)$ at that instant. We can then see that geometrically $\dot{c}(t):=\frac{d}{ds} c(s)|_{s=t} \in T_c \conf$, and this allows to define the following map:
\[
E_c: T_c \conf \rightarrow \Gamma(T \spa) ; v \mapsto (x \mapsto (x, (v \circ c^{-1})(x))
\]
\noindent which associates to a configuration $c$ and its velocity $\dot{c}$ a complete vector field in the space $c(\matter)$. Due to the hypothesis of the bijective nature of elements in $\conf$ we can than also take the pullback $c^*$ of a vector field and we have the complete picture relating the Lagrangian (configuration) velocity to the Eulerian vectorfield (space) and the Convective vector field (matter):
\begin{center}
\begin{tikzcd}
\underbrace{\Gamma(T \matter)}_{\text{convective}} &  \arrow [l, "(c^* \circ E_c)"'] \underbrace{T_c \conf}_{\text{Lagrangian}} \arrow[r,"E_c"] &  \underbrace{\Gamma(T \spa)}_{\text{Eulerian}}
\end{tikzcd}
\end{center}

\section{The Principle Bundle structure of continuum motion \label{sec:princbundle}}
We can now consider the manifold $\text{Diff}(\spa)$ of diffeomorphic auto-morphisms of $\spa$  and considering that $\spa$ has a metric, we can consider the isometries within $\text{Diff}(\spa)$ which will be indicated as $G \subset \text{Diff}(\spa)$. It is possible to see that $G$ is a Lie group and in the case that $(\spa,g)$ is flat and therefore Eucledian, $G$ would correspond to the Special Eucledian group $SE(n)$. Keeping it general, we can consider $G$ as the group of isometric right actions on $\spa$:
\[
\lhd: \spa \times G \rightarrow \spa : (p,h) \mapsto p \lhd h
\]
\noindent with the corresponding infinitesimal generator of the corresponding Lie algebra $\mathfrak{g}$ for which $
\forall T \in \mathfrak{g}$ we can define:
\[
V_T:\spa  \rightarrow T\spa : p, \mapsto \frac{d}{ds}( p \lhd e^{sT})|_{s=0}
\]
\noindent for which the isometric property should be satisfied $\forall T \in \mathfrak{g}$:
\[
\mathcal{L}_{V_T} g=0.
\]
\noindent We can also define an equivalent relation in $\conf$. Considering $c_1,c_2 \in \conf $:
\[
c_1 \sim c_2 
\Leftrightarrow \exists h \in G \; \text{s.t.} \; c_1(\matter)=c_2(\matter) \lhd h
\]
\noindent  and it is straight forward to see that such an equivalence relation satisfies the needed properties for an equivalence relation. We can then define the {\it pure deformation space} as $\deform:=\conf/\sim$ which can be recognised as the space of orbits for the group $G$. 
This shows that $\conf$ has the following fiber bundle structure:
\[
\pi:\conf \rightarrow \deform : c \rightarrow [c]
\]
\noindent where $[c]$ indicates the equivalence class in $\deform$ 
to which $c$ belongs. It is possible to finally check that the right action $\lhd$ is free and therefore $\conf \rightarrow \deform$ is a principle bundle on which $G$ acts.


By the previous construction it is possible to see if a motion expressed by an element $v_c(c) \in T_c \conf$ is a pure rigid body motion by observing if it lies in the vertical space $\text{Ver}(c):=\text{Ker}(\pi_*(c))$, the kernel of $\pi_*$ in $c$, which is isomorphic to $\mathfrak{g}$. On the other hand, without a horizontal space $\text{Hor}(c)$ complementary to $\text{Ver}(c)$ such that we have $T_c \conf = \text{Ver}(c) \oplus \text{Hor}(c) \; \forall c \in \conf$, it is not possible to uniquely decompose any general vector in a vertical and horizontal component . For a principal bundle, the definition of a horizontal differential distribution $\text{Hor}(c)$ is done with a connection which is a Lie algebra valued (in our case $\mathfrak{g}$) one form. As it will be shown in this paper it is possible to define such a connection in a very precise way, but before doing so, we will need to restrict to the case of the 3 dimensional space and to the Lie group of positive isometries $SE(3)$, the Special Eucledian group with corresponding Lie algebra $\mathfrak{se(3)}$.

\section{The geometry of $\mathfrak{se(3)}$ and screws \label{sec:screws}}
Without the intention to be exhaustive about a theory which is well known and documented \cite{Ball:1900, book:selig, StramigioliTutorial2001}, a short introduction of screw theory will now be given.

In the modeling of rigid body motions, the role of $SE(3)$ is fundamental because it is isomorphic to the configuration space of a rigid body motions as explained in the previous section. Infinitesimal rigid body motions (generators) can then be represented using elements of the corresponding Lie algebra $\mathfrak{se(3)}$ which are in mechanics called {\it instantaneous  twists}. In screw-theory \cite{Ball:1900} the kinematics of rigid-body mechanisms is described using the geometry of lines and screws in the  Euclidian space. The latter are geometric lines with the extra information of an orientation of the line, a scalar magnitude and an extra scalar quantity $\lambda$ called the pitch of the screw. This is based on Mozzi's theorem which basically says that any rigid body instantaneous motion can be expressed as a rotation around an axis in the Euclidean space, plus a translation along the same axis.
The magnitude of the screw will then indicate the speed of rotation around this axis.
The relation between the rotation speed and the translation speed is given by what is called the {\it pitch} as it happens in a mechanical screw. As shown in \cite{JournalBall2002, book:selig}, the description of rigid body motions using Lie groups and screws are equivalent. This means that we can properly interpret an abstract element of the Lie algebra $\mathfrak{se(3)}$ as a screw in the Eucledian space acted upon $SE(3)$. Due to this fact, it is geometrically meaningful to take linear combinations of screws using the vector structure of $\mathfrak{se(3)}$ and this will be one of the core insights used: after introducing the concepts of screws, we will see that it is possible to associate an infinitesimal screw to any motion of the continua and then sum/integrate in the continua.
Furthermore, being $\mathfrak{se(3)}$ not only a vector space, but also a Lie algebra, it possesses a skew symmetric operation which corresponds in screw theory to what is called the cross product of screws.

One of the advantages of the screw interpretation is that the infinitesimal motion can be associated to a geometrical object living in the same space where the motion is happening, which is not the case with the more abstract element of a Lie algebra which belongs to a different space. The screw construction is possible by first individualising the straight line which is kept invariant by the infinitesimal motion under consideration, as expressed by Mozzi's theorem, and then identifying the other scalar information which defines the screw, the pitch $\lambda$. 
\subsection{Screw representation}
We can uniquely identify a screw with one  vector $\omega \in T_p \spa$, representing rotation, and a vector $v \in T_p \spa$ \cite{StramigioliTutorial2001}. Properly speaking this geometry uses what are called line bounded vectors for the geometry of lines but this distinction is not needed for our purpose. We can then define what is called a twist as:
\begin{equation}
T:=
\begin{pmatrix}
\omega \\ v
\end{pmatrix} \label{eq:screw}
\end{equation}
\noindent As said, it can also be seen that this "double vector" composed by two three dimensional vectors is a proper representation of an element of $\mathfrak{se(3)}$ as it will be constructed hereafter, and therefore we can sum such elements \cite{Stramigioli2001Book}. In order to see this, the  first needed step is to decompose $v$ along $\omega$ ($v_\omega$) and a part belonging to the orthogonal complement ($v_\omega^\perp$) and this can be done using the inner product $g$ in the case in which $\omega \ne 0$:
\[
v=v_\omega + v_\omega^\perp.
\]
\noindent The case in which $\omega=0$ represents  the degenerate case of a ``pure translation'' which in this setting can be properly framed as a rotation with an axis at infinity using the concept of improper lines of projecting geometry \cite{phd:lipkin}.

\subsubsection{Moment of the line and screw axis}
In the non degenerate case, the information about $\omega$ and $v_{\omega}^\perp$ uniquely identifies the axis of the screw. This is done by identifying the unique axis in space for which a rotation around such an axis with speed $\omega$ would result in the velocity $v_\omega^\perp$. It is possible to see that the axis of the rotation is oriented along the $\omega$ direction and it passes through a point which is situated at 
\begin{equation}
r:=\frac{\omega \times v}{<\omega,\omega>}
\label{eq:ray}
\end{equation}
\noindent from the point $p$, the base point of the tangent space $T_p \spa$ where $\omega$ and $v$ live. The operator $\times$ in this cases indicates the vector product of vectors in 3D Euclidean space. 
It is clear that by construction, the component $v_\omega$ plays no role in the definition of $r$ because it is aligned with $\omega$ and therefore results in a vanishing vector product of parallel vectors.

\subsubsection{The Pitch of the screw}
The pitch of the screw is then defined as the unique positive scalar $\lambda$ for which:
\[
v_\omega=\lambda \omega
\]
\noindent holds. It is easy to see  that we can easily calculate $\lambda$ as:
\begin{equation}
\lambda=\frac{<\omega,v>}{<\omega,\omega>}
\label{eq:pitch}
\end{equation}
\noindent In this case it is possible to see that  $v_\omega^\perp$ plays no role because by construction it is perpendicular to $\omega$. This shows that $v_\omega$ is the component of velocity which is related to the translation along the axis of the screw and it is different than zero if and only if the pitch $\lambda$ is different than zero which is the case if and only if the motion is not a pure rotation.

It is possible to see that if the motion of the body is rigid, the previous analysis in each other point $q \ne p$ belonging to the rigid body, would result in $\bar{\omega},\bar{v} \in T_q \spa$ which would give the same axis in the Euclidean space, and the same $\lambda$. Ultimately this results in an expression of Mozzi's theorem:
\begin{equation}
    \begin{pmatrix}
     \omega \\ v
    \end{pmatrix}
    =
    \underbrace{
    \begin{pmatrix}
     \omega \\ r \times \omega
    \end{pmatrix}}_{\text{rotation}}
 +
 \begin{pmatrix}
     0 \\ \lambda \omega
    \end{pmatrix}
\end{equation}
\noindent which clearly gives an expression of the all motion as built up from $\omega$ indicating the direction and speed of the motion, $r$ indicating where the axis of the screw lies in Euclidian space, and the scalar $\lambda$. It is important to notice that $r \times \omega$ which is called the {\it moment of the line} depends on the geometrical position of the line. A point of the axis of the screw can be reached by moving along $r$ from $p$. This is a properly defined statement because of the flatness of the space and the existence of a metric which allows such a construction.

\subsubsection{Summing screws and the Adjoint representation of $SE(3)$ \label{sec:adjoint}}
In the previous construction, we have given a way to associate to a motion a screw which is identified by a line with a magnitude plus a pitch. In order to sum these objects, we could do it geometrically as it is done in screw theory by a geometric operation of lines, or use the fact that a screw is correspondent to an element of $\mathfrak{se(3)}$ and use its vector structure. For what will be explained later,  we want to be able to calculate an integral which is a generalisation of a sum and the second method using $\mathfrak{se(3)}$ is much more convenient. 

For the Lie group $SE(3)$ of positive isometries of the Euclidean space of which $\mathfrak{se(3)}$ is the algebra, we can consider what is called the Adjoint representation of the group. If we take $h,k \in SE(3)$, we can consider what is called the conjucation automorphism $\Psi_k(h):=k h k^{-1}$ and we can then define
\begin{equation}
    Ad_k:\mathfrak{se(3)} \rightarrow \mathfrak{se(3)}; T \mapsto (\Psi_k)_*T
\end{equation}
\noindent which basically describes how elements of $\mathfrak{se(3)}$ change by changing the observer by $k$. In our case this would mean how the screw changes in the Euclidean space by moving the observer by $k$ which more precisely would mean that if consider $q=k(p)$ where $p,q$ are points in the Euclidian space, $\omega_p, v_p \in T_p \spa$, representing a motion as a screw in $p$, $\omega_q, v_q \in T_q \spa$ representing the same screw (motion) in $q$, and $k \in SE(3)$, then we have that \cite{book:selig}:
\begin{equation}
\label{eq:adjoint}
   \begin{pmatrix}
    \omega_q \\ v_q
   \end{pmatrix}=
   Ad_k 
   \begin{pmatrix}
    \omega_p \\ v_p
   \end{pmatrix}
\end{equation}
\noindent Thanks to the $Ad_k$ map, it is now possible to meaningfully define sum of multiple screws by first transforming them to the same tangent space:
\begin{equation}
   \begin{pmatrix}
    \omega_q \\ v_q
   \end{pmatrix}
       =\sum_i
   Ad_{k_i} 
   \begin{pmatrix}
    \omega_{p_i} \\ v_{p_i}
   \end{pmatrix}
\end{equation}
\noindent where $p_i \in T_{p_i} \spa$ and $q=k_i(p_i)$. This will be the main construction which will be extended to integration in the next section.

\section{The Pure Deformation Connection \label{sec:infscrews}}
As explained in the previous section, the screw resulting by the previous analysis, will be the same if calculated in any moving point of a rigid body. In case of a general moving continua, this will clearly not be the case in general, but thanks to the fact that such screws, are elements of the Lie algebra $\mathfrak{se(3)}$, it is completely meaningful to sum/integrate their contributions in order to calculate the ``total screw''  of the volume which has been used for the integration. In case the motion would be rigid, this would correspond to the screw of a rigid body as explained, but in case that the motion will not be rigid, this will result in an expression of a $\mathfrak{se(3)}$ valued 1-form which paired with the vector field of motions, would result after integration in a screw. Such an object by construction is what is called a differential geometric connection which can then be used to define a horizontal space in the principal bundle. 

\subsection{Screw interpretation of flow, vorticity, entropy and helicity}
If we considered a rotating rigid disc with an angular velocity of $\omega$, it is easy to see that all points on the disc would have a tangential velocity $v$ equal to the angular velocity times their distance from the axis. Furthermore, the corresponding vorticity, which is expressed in vector calculus with the rotational $\nabla \times v$, would be constant on the all disk and equal to $2 \omega$ because, as shown in  \cite{Chorin1993AMechanics}, we have:
\[
\omega=\frac{1}{2} \nabla \times v
\]
\noindent Writing this expression using exterior calculus, we then have that
\[
\omega=\frac{1}{2} (dv^\flat)^\sharp
\]
\noindent which would mean that eq.(\ref{eq:adjoint}) would be satisfied $\forall p,q$ on the rigid body.

If for a general continum media, we consider a covariant representation of both the velocity field $v \in \Gamma(T \spa)$ as a one form field $v^\flat \in \Omega^1(\spa)$ obtained by lowering its contravariant index with the metric $g$, and then the vorticity as its  exterior derivative $d v^\flat \in \Omega^2(\spa)$, we can then transform the two-form $d v^\flat$ to the one-form $ \star d v^\flat$ using the Hodge star operator $\star$ based on the metric $g$ and finally rise the index using the metric $g$ which we indicate with the standard $\sharp$ notation.
We can then see with reference to (\ref{eq:screw}) that the infinitesimal (to be integrated on the volume) screw can then be defined as:
\begin{equation}
\delta T:=
\begin{pmatrix}
\frac{1}{2}(\star d v^{\flat})^\sharp\\
v
\end{pmatrix}.
\label{eq:inftwist}
\end{equation}

\noindent This is because, as just stated, you have a clear correspondence $\omega \leftrightarrow \frac{1}{2}(\star d v^\flat)^\sharp$ of the vorticity as the local angular velocity.

A number of new concepts can then be introduced, like the {\it enstropy} topform:
\begin{equation}
\label{eq:enstropy}
e:=
d v^{\flat} \wedge \star dv^\flat \in \Omega^3(\spa) 
\end{equation}
 which plays the role of measuring the norm of the vorticity once integrated on the volume,
 and the {\it helicity} topform:
\begin{equation}
\label{eq:helicity}
h:=v^{\flat} \wedge dv^\flat \in \Omega^3(\spa).
\end{equation}

\noindent With analogy to eq.(\ref{eq:ray}) we can then define the {\it ray one form}:
\begin{equation}
\label{ray}
r:=\frac{\iota_v dv^\flat}{\star e}
\in \Omega^1(\spa) 
\end{equation}
and with analogy to eq.(\ref{eq:pitch}) we can define the {\it pitch function}

\begin{equation}
\lambda:=\frac{\star h}{\star e}
\in \Omega^0(\spa) 
\label{eq:pitchfunction}
\end{equation}
\noindent As the enstropy is used as a normalisation factor corresponding to the $<\omega,\omega>$ in eq.(\ref{eq:ray}) and eq.(\ref{eq:pitch}), it is interesting to observe that the helicity is proportional to the defined pitch form which is exactly the intuition of the concept (and corresponding naming) of helicity: with no pitch, there would be no helical motion of the fluid and the helicity would vanish. 

\begin{theorem}[Screw Field Decomposition]
Given any continuum vector field $v$ with corresponding covariant representation $v^\flat$ and respective vorticity $d v^\flat$, the following decomposition holds:

\begin{equation}
    \begin{pmatrix}
     {\star dv^\flat} \\  \star v^\flat
    \end{pmatrix}
    =
    \begin{pmatrix}
            \star dv^\flat  \\   
               r \wedge \star dv ^\flat
    \end{pmatrix}
 +
 \begin{pmatrix}
     0 \\ \lambda  dv^\flat
    \end{pmatrix}
\end{equation}

\noindent 
\end{theorem}
\begin{proof}
We can start by writing the identity:
\begin{equation}
\underbrace{\star (d v^\flat \wedge \star d v^\flat)}_{\star e} v^\flat
=
\underbrace{\star (d v^\flat \wedge \star d v^\flat) v^\flat
-\star(v^\flat \wedge dv^\flat) \star d v^\flat }_{\alpha}+
\underbrace{\star (v^b \wedge d v^\flat)}_{\star h} \star d v^\flat
\end{equation}
\noindent We can then further process the term $\alpha$. Using the identity $\iota_v \alpha=\star (v^\flat \wedge \star \alpha)$ in the following we obtain:
\begin{equation}
    \alpha=\star \iota_v (dv^b \wedge \star d v^\flat)-(\iota_v \star dv^\flat)\star dv^\flat
\end{equation}
\noindent  and considering that the scalar function commute with the Hodge
\begin{equation}
 \alpha=\star(\iota_v(dv^\flat \wedge \star d v^\flat)
 -(\iota_v \star dv^\flat) dv^\flat)
\end{equation}
\noindent and using the distributive property of the $\iota_v(\alpha \wedge \beta)$ we obtain:
\begin{equation}
 \alpha=\star(\iota_v(dv^\flat) \wedge \star dv^b).
\end{equation}
\noindent Finally, dividing the original equation for the scalar $\star e$, and then taking the Hodge $\star$ on both sides of the equation the result is obtained

\end{proof}
\noindent Taking the Hodge and then the sharp of the second component and the sharp of the first component, we have an expression for the infinitesimal twist of eq.(\ref{eq:inftwist}) as:
\begin{equation}
    \begin{pmatrix}
     ({\star dv^\flat})^\sharp \\  v
    \end{pmatrix}
    =
    \begin{pmatrix}
            ({\star dv^\flat})^\sharp  \\   
               (\star(r \wedge \star dv^\flat)))^\sharp 
    \end{pmatrix}
 +
 \begin{pmatrix}
     0 \\ \lambda  (\star dv^\flat)^\sharp
    \end{pmatrix}
\end{equation}
\noindent which can be manipulated recognising being equal to 
\begin{equation}
    \begin{pmatrix}
     \frac{1}{2}({\star dv^\flat})^\sharp \\  v
    \end{pmatrix}
    =
    \begin{pmatrix}
            \frac{1}{2}({\star dv^\flat})^\sharp  \\   
               r^\sharp \times  ({\star dv^\flat})^\sharp 
    \end{pmatrix}
 +
 \begin{pmatrix}
     0 \\ \lambda  (\star dv^\flat)^\sharp.
    \end{pmatrix}
\end{equation}
\noindent This expression gives a construction to identify geometrically exactly the same screw interpretation of a screw for an infinitesimal element. It is therefore possible for a continuum to define a screw field for which we have a pitch field which is equal to the ratio of the density of the helicity and entropy as introduced in eq.(\ref{eq:pitch}) and with reference to eq.(\ref{ray})
a ray vector field as
\[
r^\sharp=
\begin{pmatrix}
 \frac{\iota_vdv^\flat}{\star e}
\end{pmatrix}^\sharp.
\]

\subsection{The Principal Bundle Connection}

It is now possible to define the integral of the various contributions of the continuum, but in order to do so, we first need to map the various contributions to a common reference as explained in sec.\ref{sec:adjoint} via the $Ad_k$ where $k$ is an element of $SE(3)$ which is defined based on the location of the point $p$ where the infinitesimal screw is considered. To clarify the construction, consider to choose an ortho-normal frame bundle ($<e_i, e_j>=\delta_{ij}$) with $e_i(p) \in T_p \spa$ such that $\text{span}\{e_i(p)\}=T_p \spa$. This can be used for example to express:
\[
v(p)=v^i(p) e_i(p) \qquad 
({\star dv^\flat})^\sharp(p)
=(({\star dv^\flat})^\sharp)^j(p) e_j(p)
\]
\noindent where Einstein implicit summation is used. If we define the six dimensional base of the screws (and therefore $\mathfrak{se(3)}$) at the point $p$ with:
\[
s_i(p):=
\begin{pmatrix}
 e_i(p) \\ 0_3(p)
\end{pmatrix} \forall i=1..3 \quad
s_j:=
\begin{pmatrix}
 0_3(p) \\ e_{(j-3)}(p)
\end{pmatrix} \forall j=4..6 \quad
\]
\noindent where with $0_3(p)$ we indicate the null vector of $T_p \spa$, we could see that, due to the fact that screws are isomorphic to elements of $\mathfrak{se(3)}$, $s_1,..s_6$ is actually a basis of $\mathfrak{se(3)}$ which is created based on the choice of the point $p$ and the $e_i(p) \in T_p \spa$. Following this reasoning, if we would choose a different point $q$ with $e_i(q) \in T_q \spa$, and then define in the same way the basis built on the new point and basis vectors, there would be an element of $k_{p\rightarrow q} \in SE(3)$ such that $s_i(q) =Ad_{k_{p \rightarrow q}} s_i(p) \; \forall i=1..6$. This means that if we have a general screw $T \in \mathfrak{se(3)}$ we can express it using a frame attached to either $p$ or $q$ using two different frames in $p$ and $q$ as just introduced, and in such a situation we can associate two numeric representations $T^i(p)$ or $T^i(q)$ such that:
\[
T=T^i(p)s_i(p)=T^j(q)s_j(q)
\]
\noindent 
We can then define for each base $s_i(p)$ the canonical dual base $s^i(p) \in \mathfrak{se^*(3)}$ such that $s^i(p)s_j(p)=\delta^i_j$. Using the canonical dual base it is then possible for any element $T \in \mathfrak{se(3)}$ and any base $s_i(p)$ to calculate the components for which $T=T^i(p)s_i(p)$ as:
\[
T^i(p)=s^i(p)(T).
\]

\subsubsection{The connection}
It is now possible to define the principal bundle connection which can be used to define the horizontal space. Consider the volume form $\mu \in \Omega^3(\spa)$ which can be directly defined using the metric $g$ as indicated previously. Consider a volume $\bar{\spa} \subset \spa$ for which we want to describe the decomposition between rigid body motion and pure deformation of the continuum media. This could be for example the image of a solid via an embedding $c$. If we consider a configuration velocity $v_c \in T_c \conf$, this would induce, as discussed previously, a vectorfield $v=E_c(v_c) \in \Gamma(T\spa)$. We can then rewrite $\frac{1}{2}(\star d v^{\flat})^\sharp$ as a composition of operations as $\frac{1}{2}( g^{-1} \circ \star \circ d \circ g)(v)$ to then define the following $\mathfrak{se(3)}$ valued one form on $\conf$:
\begin{multline}
\label{eq:connection}
    \nu_{\bar{S}}: T \conf \rightarrow \mathfrak{se(3)}: \\
    v_c \mapsto
    \frac{1}{V}
    s_i(q)
    \int_{\bar{S}}
    s^i(q)
    \begin{pmatrix}
         Ad_{k_{p \rightarrow q}}
    \begin{pmatrix}
     \frac{1}{2}( g^{-1} \circ \star \circ d \circ g) \\ \text{Id}
    \end{pmatrix}
    E_c(v_c)(p)
        \end{pmatrix}
    \otimes \mu(p)
\end{multline}
\noindent where we define the total volume:
\[
V:=\int_{\bar{\spa}} \mu
\]
\noindent What eq.(\ref{eq:connection}) actually does, is to take for each point $p$ the vector $v(p)=E_c(v_c)(p)$ and then build the infinitesimal screw 
using the form introduced in eq.(\ref{eq:inftwist}) which is a screw expressed in $p$. The result is than mapped to any common frame in $q$ where the components can be extracted via $s^i(q)$ to be then integrated. In this way six scalar integrals are calculated which represent the integral of the components of the expression of the infinitesimal screws in $q$ expressed with the bases $s_i(q)$. Outside the integral, using the bases elements, the six scalar values are then composed to the result in $\mathfrak{se(3)}$ and normalised by the volume considered. If as an example we would consider a rigid body motion of the continuum, the result of the integral would be the volume times the constant screw representing the rigid body motion in $s_i(q)$ for which the external division for the volume $V$ and multiplication for the base elements would return the abstract screw representing the rigid body motion.

\section{Implications of the presented structure \label{sec:implications}}
The connection $\nu_{\bar{S}}$ allows now to define in an intrinsic way a Horizontal space for a volume $\bar{S}$ as:
\[
\text{Hor}_{_{\bar{S}}}(c):=\{v_c \in T_c \conf \; s.t. \; \nu_{\bar{S}}(v_c)=0 \}
\]
\noindent which allows now to define a unique decomposition of the velocities in configuration space:
\[
T_c \conf = \mathfrak{se(3)} \oplus \text{Hor}_{\bar{S}}(c).
\]
\noindent It interesting to relate the introduced decomposition to an important identity:
\[
\nabla v^\flat =\frac{1}{2} \mathcal{L}_v g - \frac{1}{2} d v^\flat
\]
\noindent which should be interpreted as an equality in the $T^0_2 \spa$ tensor fields. The equality basically says that the covariant differential of the one form representation of the vector field can be decomposed in a symmetric part equal to $\frac{1}{2}L_v g$ which represents what is called the  {\it rate of strain} and an anti-symmetric part $\frac{1}{2} dv^\flat$  which represents half of  the {\it vorticity}. Locally, if the symmetric part is equal to zero, it means that the local motion is not deforming the metric and corresponds to a local rigid body motion. In such a case the all covariant differential will correspond to the vorticity which represents the infinitesimal rotation part of the screw corresponding to the first three components. This expression gives therefore via $\frac{1}{2}L_v g$ a local measure of the rate of strain and it excludes all translations which are happening as a collection of the all continuum. The proposed connection gives instead a way to decompose the global motion. Furthermore, the construction presented for the connection can be modified to use a mass top-form $\bar{\mu}_m:=c_* \mu_m \in \Omega^n(\spa)$ rather than the volume top-form $\mu$. In such a case, the integrated infinitesimal screws representing velocities, would be also multiplied by the density and the volume (the mass form) which would result in the dimension of momenta rather than velocities. This would result in the dual of $\mathfrak{se(3)}$ representing the total screw moment \cite{Stramigioli2001h} belonging to $\mathfrak{se^*(3)}$. This will be the content of a future paper showing the geodesic motion of the rigid part of a continuum.

\subsection{The pure deformation space and elastic energy}
Thanks to the recognition of the principal structure of the configuration space of a continua, it is possible to recognise the base space $\deform$ as the space expressing the pure deformations of the matter. 
Suppose to be in a configuration $c \in \conf$ which corresponds to a deformation $d=\pi(c) \in \deform$ using the projection of the bundle. For an infinitesimal motion $\delta c \in T_c \conf$ it is now possible to uniquely decompose it in an horizontal and vertical components:
\[
\delta c = \delta r + \delta d \quad \delta r \in \text{Ver}(c), \delta d \in \text{Hor}(c)
\]
\noindent in which the vertical space is six dimensional and the horizontal is infinite dimensional. If $\delta d=0$ we have a pure rigid motion and in such a case, the infinitesimal vector field $\delta v=E_c(\delta c) \in T \spa$ is such that $\mathcal{L}_{\delta v} g=0$ everywhere. In all other cases, a component $\delta d \ne 0$ will result in a motion for which $(\mathcal{L}_{\delta v} g)(p) \ne 0$ for some $p \in \spa$. Furthermore, any finite motion starting in $c_1 \in \conf$ and ending in $c_2 \in \conf$ for which the deformation is the same $(d:=\pi(c_1)=\pi(c_2))$, clearly represents configurations in which the internal deformation $d$ is the same because by construction there would be a rigid motion bringing $c_1$ to $c_2$.
This means that a completely coordinate invariant way of expressing potential energy of the deformation can be described by defining an energy function:
\begin{equation}
    H:\deform \rightarrow \mathbb{R} \label{eq:energy}
\end{equation}
\noindent which associates to any internal state $d \in \deform$ of the continua an energy function. Remembering  that $\deform$ is an infinite dimensional manifold and an expression for which $H$ can be specified can only be given via an integral expression whose variational derivative would represent the dual stress field to the deformation represented by the variation. To make this statement more precise, we can use the Lie derivative identity of the pull back:
\[
c^*(\mathcal{L}_vg)=\mathcal{L}_{(c^*v)}c^*g
\]
\noindent which gives an expression of the pullback of the rate of strain field in $\spa$ to the rate of strain field in $\matter$. Considering the non-singularity of the pull back $c^*$ due to the diffeomorphic hypothesis and the positive definiteness of $g$, this implies that $c^*(L_vg)$ is an equivalent description on $\matter$ of the rate of strain $\mathcal{L}_vg$ on $\spa$. Furthermore, a material infinitesimal deformation $\delta d \in T_c \conf$ can be associated uniquely to an infinitesimal vector field $\delta \bar{d}=E_c(\delta d) \in T \spa$ and to an infinitesimal vector field $c^* (\delta \bar{d})  \in T \matter.$ This shows that, given a configuration $c \in \conf$, and a deformation $\delta d$, we can uniquely identify $\mathcal{L}_{c^* \delta \bar{d}} c^* g$. Furthermore,  $\mathcal{L}_{c^* \delta \bar{d}} c^* g$ can be recognised being geometrically part of the tangent space of $\text{sym}_+(T^0_2(\matter))$, the space of positive definite $(0,2)$ tensors of metrics
at $c^*g$. This shows therefore that
\[
\deform \simeq \Gamma(\text{sym}_+(T^0_2(\matter)))
\]
\noindent and a curve in $\deform$ resulting in a $\mathcal{L}_{v} g \ne 0$ will result in a corresponding $\mathcal{L}_{(c^* \delta \bar{d})} g_m \ne 0$ which 
will change the $\text{sym}_+(T^0_2(\matter))$ tensor field $g_m:=c^*g$. This gives a constructive way to define the energy function (\ref{eq:energy}) as:
\begin{equation}
    H:\Gamma(\text{sym}_+(T^0_2(\matter))) \rightarrow \mathbb{R} : \bar{g} \mapsto \int_{\matter} \mathfrak{H}(\bar{g})
\end{equation}
\noindent where for each point $m \in \matter$, $\mathfrak{H}$ associates to $\bar{g}(m)=c^*(g)(m) \in \text{sym}_+(T^0_2(\matter))$ a value $\mathfrak{H}(\bar{g})(m) \in \Omega^n_p(\matter)$.
It is now possible to take the variational derivative of $H$ by considering the following differential expression:
\begin{equation}
H(g+\delta g)-H(g)=<\frac{\delta H}{\delta g} | \delta g> \label{eq:varder} + O^2(\delta g)
\end{equation}
where $<|>$ is a proper dual pairing. In order to define this pairing in a meaningful, geometric way, we can first observe that being $\delta g$ a symmetric $(0,2)$ tensor, if we contract one of its indices with the inverse of the metric $g$ we would get a linear operator:
\[
\delta \bar{g}_i^j=(g^{-1})^{jl}\delta g_{li}.
\]
\noindent Being $\delta \bar{g}$ a liner operator we can do two things: first we can consider the eigen-spaces of $\delta \bar{g}$ which will indicate the principal direction of deformations, and second, we can interpret it as a vector valued one form:
\[
\delta \bar{g}(m) \in T^1_m \matter \otimes \Omega^1_m(\matter) \quad m \in \matter .
\]
\noindent This last consideration gives a very useful expression of the variational derivative of eq.(\ref{eq:varder}):
\begin{equation}
H(g+\delta g)-H(g)=
\int_{\matter} \mathfrak{T} \dot{\wedge} \delta \bar{g} \label{eq:varder1}
\end{equation}
where $\mathfrak{T} \in \Omega^1(\matter) \otimes \Omega^{n-1}(\matter)$ is a one form-valued $(n-1)$ form which can be naturally paired using $\dot{\wedge}$ with $\delta \bar{g}$ and is defined as:
\[
a_1 \otimes a_2 \in
T^1\matter \otimes \Omega^1(\matter) 
, \qquad 
b_1 \otimes b_2 \in
 \Omega^1(\matter) \otimes \Omega^{n-2}(\matter), 
\]
\[
\Rightarrow
(a_1 \otimes a_2) \dot{\wedge} (b_1 \otimes b_2):=b_1(a_1) a_2 \wedge b_2 \in \Omega^n(\matter)
\]
\noindent which can be therefore integrated on the all matter space to give a scalar expression of the infinitesimal energy variation due to the deformation $\delta \bar{g}$. $\mathfrak{T}$ is the stress which is a direct consequence of the definition of the energy function $H$ and is naturally a one form (representing an infinitesimal force) valued $(n-1)$ form (representing and infinitesimal surface), which can associate to any infinitesimal surface the corresponding infinitesimal force.

\section{Conclusions \label{sec:conclusions}}
In this paper it has been shown that the infinite dimensional space of configurations of a continuum has the structure of a principle bundle for which the base space $\deform=\text{sym}_+(T_2^0(\matter))$ of pure deformations is the space of sections of the metric on the matter space on which a potential energy function can be defined intrinsically to describe nonlinear elasticity. 

Furthermore, it has been shown that from the velocity vector field together with its vorticity a $\mathfrak{se(3)}$ valued vector field of screws can be defined which, if properly integrated can be used to define a geometric connection which allows to decompose the vector field in a rigid body motion and a pure deformation. Last but not least the relevance of concepts like helicity and enstropy have been defined and shown during this construction.

This work is an important step in defining a formulation of open and interconnectable nonlinear elasticity which will be needed to create open models of deformable bodies as needed in the ERC project portwings (www.portwings.eu) 

\section*{Acknowledgement}
The author would like to acknowledge the financial support of the European commission ERC Grant agreement ID: 787675 of the project www.Portwings.eu and the fruitful discussion and feedback from Dr. Federico Califano and Dr. Ramy Rashad Hashem.

\bibliographystyle{elsarticle-num-names} 
\bibliography{references.bib}

\end{document}